\documentclass[a4paper,12pt,reqno]{amsart}

\usepackage{amsfonts}
\usepackage{amsmath}
\usepackage{amssymb}
\usepackage{cite}
\usepackage{mathrsfs}
\usepackage{enumitem}
\usepackage[colorlinks]{hyperref}
\usepackage{bm}
\usepackage{comment}

\numberwithin{equation}{section}

\usepackage{graphicx}

\newtheorem{theorem}{Theorem}[section]
\theoremstyle{definition}
\newtheorem{definition}[theorem]{Definition}
\newtheorem{remark}[theorem]{Remark}
\newtheorem{corollary}[theorem]{Corollary}
\newtheorem{proposition}[theorem]{Proposition}
\newtheorem{lemma}[theorem]{Lemma}
\newtheorem{example}[theorem]{Example}

\newcommand{\Tr}{\textnormal{Tr}}

\begin{document}

\title{$\mathbb{F}$-valued trace of a finite-dimensional commutative $\mathbb{F}$-algebra}
\author[Anuj Kumar Bhagat]{A. K. Bhagat}
\author[Ritumoni Sarma]{R. Sarma}
\address{
	Anuj Kumar Bhagat:
	\endgraf
	Department of Mathematics
	\endgraf
	Indian Institute of Technology, Delhi, Hauz Khas
	\endgraf
	New Delhi-110016 
	\endgraf
	India
	\endgraf
	{\it E-mail address} {\rm anujkumarbhagat632@gmail.com}
}

\address{
	Ritumoni Sarma:
	\endgraf
	Department of Mathematics
	\endgraf
	Indian Institute of Technology, Delhi, Hauz Khas
	\endgraf
	New Delhi-110016 
	\endgraf
	India
	\endgraf
	{\it E-mail address} {\rm ritumoni407@gmail.com}
}
\begin{abstract}
A non-zero $\mathbb{F}$-valued $\mathbb{F}$-linear map on a finite dimensional $\mathbb{F}$-algebra is called an $\mathbb{F}$-valued trace if its kernel does not contain any non-zero ideals. However, given an $\mathbb{F}$-algebra such a map may not always exist. We find an infinite class of finite-dimensional commutative $\mathbb{F}$-algebras which admit an $\mathbb{F}$-valued trace. In fact, in these cases, we explicitly construct a trace map. The existence of an $\mathbb{F}$-valued trace on a finite dimensional commutative $\mathbb{F}$-algebra induces a non-degenerate bilinear form on the $\mathbb{F}$-algebra which may be helpful both theoretically and computationally. In this article, we suggest a couple of applications of an $\mathbb{F}$-valued trace map of an $\mathbb{F}$-algebra to algebraic coding theory.
\end{abstract}
\maketitle
\section{Introduction}\label{Section 1}
Throughout this manuscript, $\mathbb{F}$ denotes a field and $\mathbb{F}_q$ denotes the finite field of order $q,$ where $q$ is a prime power.
\par
The $\mathbb{F}$-valued trace map of an extension of $\mathbb{F}$ plays an important role in the theory of both finite and infinite fields. For instance, every functional of a finite extension $\mathbb{K}$ over $\mathbb{F}$ can be described as $\alpha\mapsto \Tr_{\mathbb{K}/\mathbb{F}}(\alpha\beta),$ for a unique $\beta\in \mathbb{K}.$ In fact, $\Tr_{\mathbb{K}/\mathbb{F}}$ induces a non-degenerate bilinear form, namely, $(\alpha,\beta)\mapsto \Tr_{\mathbb{K}/\mathbb{F}}(\alpha\beta)$ on $\mathbb{K}$ and consequently, given any basis $\{\alpha_1< \alpha_2< \dots< \alpha_n\}$ for $\mathbb{K}$ over $\mathbb{F},$ there is a unique basis $\{\beta_1< \beta_2< \dots< \beta_n\}$ for $\mathbb{K}$ over $\mathbb{F}$ such that 
$\Tr_{\mathbb{K}/\mathbb{F}}(\alpha_i\beta_j)=\delta_{ij}.$ Because of these properties, the trace map turns out to be an important tool in the theory of fields. As its application, in coding theory, the trace map is used to construct ``trace codes" and it also helps in the computation of ``subfield codes"\cite{ding2019subfield}.
\par
Motivated by the wide range of applications of the trace map, we got curious to find if there is an analogue to the trace map for a finite-dimensional $\mathbb{F}$-algebra. In this manuscript, we study the trace maps of finite dimensional commutative $\mathbb{F}$-algebras. For a commutative ring $R$ with unity, and an $R$-algebra $A,$ a surjective $R$-linear map $\tau: A\to R$ is called a $R$-valued trace if $\ker(\tau)$ contains no non-zero left ideals of $A$ (see \cite{generalized}, \cite{liu2018two}, \cite{lu2020linear}). Not every finite-dimensional $\mathbb{F}$-algebra has an $\mathbb{F}$-valued trace. For example, $\mathbb{F}_2[u,v]/\langle u^2, v^2, uv\rangle$ does not admit any $\mathbb{F}_2$-valued trace. In this article, we show by construction that a finite-dimensional commutative $\mathbb{F}$-algebra of the form $\mathcal{R}=\mathbb{F}[x_1, x_2,\dots, x_n]/\langle g_1(x_1), g_2(x_2),\dots, g_n(x_n)\rangle,$ $g_i(x_i)\in \mathbb{F}[x_i],$ has an $\mathbb{F}$-valued trace.
\par
Let $\mathbb{K}$ be a finite field extension of $\mathbb{F}$, then $\mathbb{K}$ can be viewed as a vector space over $\mathbb{F}$. The multiplication by $\alpha\in \mathbb{K},$ $\mathfrak{m}_{\alpha}:\mathbb{K}\to \mathbb{K}$ given by $x\mapsto\alpha x$ is an $\mathbb{F}$-linear transformation. Then $\Tr_{\mathbb{K}/\mathbb{F}}(\alpha)$ is defined as the trace of the linear operator $\mathfrak{m}_{\alpha}.$ One may try to extend this definition to finite-dimensional $\mathbb{F}$-algebras as they are $\mathbb{F}$-vector spaces. But, sometimes this map turns out to be the zero map. For instance, if $\mathcal{R}=\mathbb{F}_2[x]/\langle x^2\rangle,$ then $\Tr_{\mathcal{R}/\mathbb{F}_2}(a+bx+\langle x^2\rangle)=0,$ for all $a,b\in\mathbb{F}_2.$ Thus, the obvious generalization of trace does not lead to a map that possesses the expected properties of the trace. However, in the case of $\mathcal{R}=\mathbb{F}_2[x]/\langle x^2\rangle,$ $a+bx+\langle x^2\rangle\mapsto a+b$ satisfies the desired properties of a trace map.
\par
In recent years, linear codes over certain $\mathbb{F}_q$-algebras(see \cite{sagar2022certain}, \cite{ZHU20112677}, \cite{shi2017constacyclic}, \cite{gao2015some}) are studied and from them, by means of Gray maps, one produces codes over $\mathbb{F}_q.$ Other means of obtaining codes over $\mathbb{F}_q$ from codes over an $\mathbb{F}_q$-algebra could be using a trace map, for instance, the trace code. Moreover, the subfield codes can be determined with the help of an $\mathbb{F}_q$-valued trace map, if exists, from codes over $\mathbb{F}_q$-algebras.
\par
The remaining sections are arranged as follows. Preliminaries are provided in Section \ref{Section 2}. The main section is Section \ref{Section 3}, where we show the existence and construction of an $\mathbb{F}$-valued trace of $\mathcal{R}=\mathbb{F}[x_1, x_2,\dots, x_n]/\langle g_1(x_1), g_2(x_2),\dots, g_n(x_n)\rangle,$ $g_i(x_i)\in \mathbb{F}[x_i].$ Certain applications of $\mathbb{F}$-valued trace are discussed in Section \ref{Section 4} before we conclude the article in Section \ref{Section 5}.

\section{Preliminaries}\label{Section 2}
\begin{definition}
    A ring $\mathcal{R}$ with unity $1_{\mathcal{R}}$ is called an \textit{$\mathbb{F}$-algebra} if there is a ring homomorphism from $\mathbb{F}$ to $\mathcal{R}$ such that $1_{\mathbb{F}}\mapsto 1_{\mathcal{R}}$ and the image of $\mathbb{F}$ is contained in the center of $\mathcal{R}.$
\end{definition}
If $\mathcal{R}$ is an $\mathbb{F}$-algebra, then $\mathcal{R}$ is a vector space over $\mathbb{F}$ and if $\dim_\mathbb{F}(\mathcal{R})<\infty,$ then $\mathcal{R}$ is called a \textit{\emph{finite-dimensional}} $\mathbb{F}$-algebra.
\begin{definition}
    Let $\mathcal{R}$ be a commutative $\mathbb{F}$-algebra and let  $S=\{\bm{r}_1,\dots, \bm{r}_n\}\subseteq \mathcal{R}.$ If there exists a surjective ring homomorphism $\mathbb{F}[x_1,\dots, x_n]\to \mathcal{R}$ such that $x_i\mapsto \bm{r}_i,$ then  $S$ is called a set of generators of $\mathcal{R}$ over $\mathbb{F}$. In this case, we say that $\mathcal{R}$ is \textit{finitely generated} $\mathbb{F}$-algebra.
    \end{definition}
    Obviously, finite-dimensional $\mathbb{F}$-algebras are finitely generated.
\begin{proposition}\label{equivalent conditions for trace}
    Let $\mathcal{R}$ be a finite-dimensional commutative $\mathbb{F}$-algebra and let $\tau:\mathcal{R}\to \mathbb{F}$ be a non-zero $\mathbb{F}$-linear map. Then the following statements are equivalent.
    \begin{enumerate}
        \item[(a)] $\ker(\tau)$ does not contain any non-zero ideals of $\mathcal{R}.$
        \item[(b)] $f:\mathcal{R}\times \mathcal{R}\to \mathbb{F}$ defined by $f(\bm{x},\bm{y})=\tau(\bm{xy})$ is a non-degenerate bilinear form, that is, $\tau(\bm{xy})=0, \forall\, \bm{x}\in \mathcal{R} \implies \bm{y}=\bm{0}.$
    \end{enumerate}
\end{proposition}
\begin{definition}\label{definition of trace}
    Let $\mathcal{R}$ be a finite-dimensional commutative $\mathbb{F}$-algebra. A non-zero $\mathbb{F}$-linear map $\tau:\mathcal{R}\to \mathbb{F}$ is called an \textit{$\mathbb{F}$-valued trace} of $\mathcal{R},$ if it satisfies any one of the statements of Proposition \ref{equivalent conditions for trace}.
\end{definition}
\begin{remark}
    In \cite{generalized}, the notion of trace (or Generalized Frobenius trace) is defined for a general ring. Let $S$ be a ring (need not be commutative) with unity and $R$ be a subring of $S$ sharing the same unity. Then a homomorphism of left $R$-modules $\Tr_R^S: S\to R$ is called a trace from $S$ to $R$ if it is surjective and $\ker(\Tr_R^S)$ does not contain any non-zero left ideals of $S.$ We note that if $R=\mathbb{F}$ and $S$ is an $\mathbb{F}$-algebra, then Definition \ref{definition of trace} is a special case of the definition in \cite{generalized}.
\end{remark}
Let $V$ be a finite-dimensional vector space over $\mathbb{F}$ and let $T$ be a linear operator on $V.$ 
\begin{definition}\cite{hoffman1971linear}
    If $\bm{\alpha}$ is a vector in $V$, then the subspace $Z(\bm{\alpha};T):=\{g(T)\bm{\alpha}: g(x)\in\mathbb{F}[x]\}$ is called the \textit{$T$-cyclic subspace generated by $\bm{\alpha}.$} If $Z(\bm{\alpha};T)=V,$ then $\bm{\alpha}$ is called a \textit{cyclic vector} for $T.$ 
\end{definition}
\begin{definition}\cite{hoffman1971linear}
    If $\bm{\alpha}$ is a vector in $V$, then then ideal $M(\bm{\alpha}; T):=\{g(x)\in \mathbb{F}[x]: g(T)\bm{\alpha}=0\}$ of $\mathbb{F}[x]$ is called the \textit{$T$-annihilator of $\bm{\alpha}.$}
\end{definition}
Since $\mathbb{F}[x]$ is an Euclidean domain, there is a unique monic polynomial $p_{\bm{\alpha}}(x)$ which generates  $M(\bm{\alpha}; T)$; the polynomial $p_{\bm{\alpha}}(x)$ is also called the $T$-annihilator of $\bm{\alpha}$. With these notations we note the following result.

\begin{theorem}[\cite{hoffman1971linear}]
    Let $\bm{\alpha}\in V\setminus\{\bm{0}\}.$ Then
    \begin{enumerate}
        \item $\dim_{\mathbb{F}} Z(\bm{\alpha};T)=\deg p_{\bm{\alpha}}(x)$
   
        \item If $\deg p_{\bm{\alpha}}(x)=k,$ then $\{\bm{\alpha}, T\bm{\alpha},\dots T^{k-1}\bm{\alpha} \}$ is a basis for $Z(\bm{\alpha};T).$ 
    \end{enumerate}
\end{theorem}
The following result that describes a basis of a tensor product of vector spaces is helpful for computation.
\begin{proposition}
    Let $V$ and $W$ be finite-dimensional vector spaces over $\mathbb{F}$ with bases $\{\bm{\alpha}_1\dots,\bm{\alpha}_m\}$ and $\{\bm{\beta}_1\dots,\bm{\beta}_n\}$ respectively. Then $\{\bm{\alpha}_i\otimes\bm{\beta}_j:1\leq i\leq m, 1\leq j\leq n\}$ is a basis for $V\otimes_{\mathbb{F}}W.$
\end{proposition}

\begin{theorem}\label{tensorproduct}
    Let, for $1\le i\le n,$ $g_i(x)\in\mathbb{F}[x]$. Then  
    $$\frac{\mathbb{F}[x_1, x_2,\dots, x_n]}{\langle g_1(x_1), g_2(x_2),\dots, g_n(x_n)\rangle} \cong \frac{\mathbb{F}[x_1]}{\langle g_1(x_1)\rangle}\otimes_{\mathbb{F}}\frac{\mathbb{F}[x_2]}{\langle g_2(x_2)\rangle}\otimes_{\mathbb{F}}\cdots\otimes_{\mathbb{F}}\frac{\mathbb{F}[x_n]}{\langle g_n(x_n)\rangle}.$$
\end{theorem}
\section{Trace on finite dimensional commutative $\mathbb{F}$-algebras}\label{Section 3}
 
\begin{theorem}\label{maintheorem}
    Let $g(x)\in \mathbb{F}[x]$ be an irreducible polynomial over $\mathbb{F}$ of degree $n$, and let $T:\mathbb{F}[x]/\langle g(x)\rangle\to \mathbb{F}$ be a non-zero functional. Set $\mathcal{R}=\mathbb{F}[x]/\langle g(x)^r\rangle,$ for $r\in \mathbb{N}.$ Then the map $\tau:\mathcal{R}\to \mathbb{F}$ given by 
        $$p_0(x)+p_1(x)g(x)+\dots+p_{r-1}(x)g(x)^{r-1} +\langle g(x)^r\rangle\mapsto \sum_{i=0}^{r-1}T(p_i(x)+\langle g(x)\rangle)$$  
    is an $\mathbb{F}$-valued trace of $\mathcal{R}.$
\end{theorem}
\begin{proof} Note that $\mathcal{B}:=\{\bm{e}_{i,j}=x^ig(x)^j: 0\leq i\leq n-1, 0\leq j\leq r-1\}$ is a basis for $\mathcal{R}$ and hence, every element $p(x)+\langle g(x)^r\rangle$ of $\mathcal{R}$ can be uniquely expressed in the form $p_0(x)+p_1(x)g(x)+\cdots+p_{r-1}(x)g(x)^{r-1}+\langle g(x)^r\rangle,$ where $p_i(x)\in \mathbb{F}[x]$ with $\deg p_i(x)\leq n-1.$ For convenience, identify $\mathbb{F}[x]/\langle g(x)\rangle$ with $\mathbb{F}^n$ via the map $f(x)+\langle g(x)\rangle\mapsto \bm{f}=(a_0,a_1,\dots,a_{n-1})\in \mathbb{F}^n,$ if $f(x)=a_0+a_1x+\cdots+a_{n-1}x^{n-1}.$ In turn, get the identification of $\mathcal{R}$ with $\mathbb{F}^{nr}$ as $\mathbb{F}$-vector spaces via the vector space isomorphism
\begin{equation}\label{isomorphism}
    p(x)+\langle g(x)^r\rangle\mapsto\overline{\bm{p}}=(\bm{p}_0, \bm{p}_1,\dots, \bm{p}_{r-1})
\end{equation}
where $\bm{p}_i=(p_{i,0}, p_{i,1},\dots,p_{i,(n-1)})\in \mathbb{F}^n.$
\par
Since $T$ is a non-zero functional on $\mathbb{F}[x]/\langle g(x)\rangle,$ there exists $\bm{s}=(s_0,\dots,s_{n-1})\in \mathbb{F}^n\setminus\{\bm{0}\}$ such that $T(\bm{a})=\bm{s}\cdot\bm{a},$ the usual dot product of $\bm{s}$ and $\bm{a}.$ With the above identification, $\tau$ can be expressed as
\begin{equation}
    \tau(\overline{\bm{p}})=\sum_{i=0}^{r-1}\bm{s}\cdot \bm{p}_i.
\end{equation}
Next consider the operator
    \begin{align*}
        \mathfrak{m}_x: \mathcal{R}&\to \mathcal{R}\\
        p(x)+\langle g(x)^r\rangle&\mapsto xp(x)+\langle g(x)^r\rangle.
    \end{align*}
    Order $\mathcal{B}$ such that $\bm{e}_{i,j}<\bm{e}_{k,l}$ if either $j=l$ and $i<k$  or $j<l.$   Suppose $A=[\mathfrak{m}_x]_{\mathcal{B}}$ is the matrix of $\mathfrak{m}_x$ relative to the ordered basis $\mathcal{B}.$ Now, for $1\leq j\leq r-1,$ compute $g(x)^j\left(p(x)+\langle g(x)^r\rangle\right)$. Observe, 
    \begin{align*}
        g(x)^j\left(\underset{i=0}{\overset{r-1}{\sum}}p_{i}(x)g(x)^{i}\right)&=\underset{i=j}{\overset{r-1}{\sum}}p_{i-j}(x)g(x)^{i}
        \end{align*}
        and
        \begin{align*}
        g(x)^r\left(\underset{i=0}{\overset{r-1}{\sum}}p_{i}(x)g(x)^{i}\right)&=\bm{0}.
        \end{align*}
    Hence,
    \[
  g(A)=
  \left( {\begin{array}{ccccc}
    O & O & \cdots & O & O\\
    I & O & \cdots & O & O\\
    O & I & \cdots & O & O\\
    \vdots & \vdots & \ddots & \vdots & \vdots\\
    O & O & \cdots & I & O\\
  \end{array} } \right)_{nr\times nr}
\]
    where $O$ is the null matrix of order $n\times n$ and $I$ is the identity matrix of order $n\times n$ over $\mathbb{F}.$ Observe that $r$ is the least positive integer such that $g(A)^r=O.$ Hence, the minimal polynomial of $\mathfrak{m}_x$ divides $g(x)^r.$ Since $g(x)$ is irreducible and $g(x)^i$ does not annihilate $\mathfrak{m}_x$ for $0\leq i\leq r-1,$ the minimal polynomial of $\mathfrak{m}_x$ is $g(x)^r.$
    \par
    Consider $\overline{\bm{s}}:=(\bm{s}, \bm{s},\dots, \bm{s})\in \mathbb{F}^{nr},$ where $\bm{s}$ is identified with $s_0+s_1x+\dots +s_{n-1}x^{n-1}\in \mathbb{F}[x]/\langle g(x) \rangle.$ Since $\mathfrak{m}_x$-annihilator of $\overline{\bm{s}}$ is $g(x)^r$ and $\deg g(x)^r=nr=\dim_{\mathbb{F}} \mathcal{R},$ the vector $\overline{\bm{s}}$ is a cyclic vector for $\mathfrak{m}_x$ and $\mathcal{B}^{'}=\{\overline{\bm{s}}^t, A\overline{\bm{s}}^t,\dots, A^{nr-1}\overline{\bm{s}}^t\}$ is a basis for $\mathcal{R}$ via the identification given in (\ref{isomorphism}).
    \par
    Now,
    \begin{equation*}
        \tau\left(x^j(p(x)+\langle g(x)^r\rangle)\right)=\tau(A^j\overline{\bm{p}})=\overline{\bm{s}}\cdot \left(A^j\overline{\bm{p}}\right)=\left(\overline{\bm{s}}A^j\right)\cdot \overline{\bm{p}}=({A^j}^t\overline{\bm{s}}^t)^t\cdot \overline{\bm{p}}.
    \end{equation*}
     If $\tau\left(x^j(p(x)+\langle g(x)^r\rangle)\right)=0, \forall\,0\leq j\leq nr-1,$ then $({A^j}^t\overline{\bm{s}}^t)^t\cdot \overline{\bm{p}}=0, \forall\,0\leq j\leq nr-1, \,\text{we have}\,\, \overline{\bm{p}}=\overline{\bm{0}}$ as $\mathcal{B}^{'}$ is a basis for $\mathcal{R}$ and $\overline{\bm{s}}$ is also a cyclic vector for $A^t.$ 
\end{proof}
\begin{remark}
The usual trace map $T=\Tr_{\mathbb{F}[\alpha]/\mathbb{F}}$, where $\alpha=x+\langle g(x)\rangle$ is always a choice of $T$ in Theorem \ref{maintheorem}.
\end{remark}
\begin{lemma}\label{trace on direct sum}
    Let $\mathcal{R}_1, \mathcal{R}_2,\dots,\mathcal{R}_s$ be finite-dimensional commutative $\mathbb{F}$-algebras. If each $\mathcal{R}_i$ admits an $\mathbb{F}$-valued trace, then so does $\prod_{i=1}^{s}\mathcal{R}_i.$
\end{lemma}
\begin{proof} 
    Suppose for each $i,$ $\tau_i:\mathcal{R}_i\to \mathbb{F}$ is an $\mathbb{F}$-valued trace of $\mathcal{R}_i$. Set $\mathcal{R}=\prod_{i=1}^{s}\mathcal{R}_i$ and define
    \begin{align*}
        \tau:\mathcal{R}&\to \mathbb{F}\\
        \bm{x}=(x_1,\dots,x_s)&\mapsto \sum_{i=1}^{s}\tau_i(x_i).
    \end{align*}
    Then $\tau$ is clearly a non-zero $\mathbb{F}$-linear map from $\mathcal{R}$ to $\mathbb{F}.$ Suppose $\tau(\bm{x})=0,$ where $\bm{x}=(x_1,\dots,x_s)\in \mathcal{R}\setminus\{\bm{0}\}.$ Without loss of generality, assume that $x_1\neq0.$ Then there is $y_{x_1}\in \mathcal{R}_1\setminus\{0\}$ such that $\tau_1(x_1y_{x_1})\neq 0.$ If $\bm{y}=(y_{x_1},0,\dots,0)\in \mathcal{R},$ then 
    \begin{equation*}
        \tau(\bm{xy})=\tau(x_1y_{x_1},0,\dots,0)=\tau_1(x_1y_{x_1})\neq0.
    \end{equation*}
\end{proof}
\begin{theorem}\label{existence of trace}
    Let $\mathcal{R}$ be a finite-dimensional commutative $\mathbb{F}$-algebra generated by only one element. Then there is an $\mathbb{F}$-valued trace of $\mathcal{R}.$
\end{theorem}
\begin{proof}
    Since $\mathcal{R}$ is a finite-dimensional commutative $\mathbb{F}$-algebra generated by one element, $\mathcal{R}\cong \mathbb{F}[x]/\langle h(x) \rangle$ for some $h(x)\in \mathbb{F}[x]\setminus\{\bm{0}\}.$ Let $h(x)=\prod_{i=1}^{s}h_i(x)^{r_i}$ be the irreducible factorization of $h(x)$ in $\mathbb{F}[x].$ Then by Chinese Remainder Theorem,
    $$
    \mathcal{R}\cong \prod_{i=1}^{s} \mathcal{R}_i
    $$
    where $\mathcal{R}_i=\mathbb{F}[x]/\langle h_i(x)^{r_i} \rangle.$ By Theorem \ref{maintheorem}, for each $i$, there exists an $\mathbb{F}$-valued trace of $\mathcal{R}_i$ and consequently, by Lemma \ref{trace on direct sum} there is an $\mathbb{F}$-valued trace of $\mathcal{R}$.
\end{proof}
However, if $\mathcal{R}$ is a finite-dimensional commutative $\mathbb{F}$-algebra generated by two elements, then there may not exist any $\mathbb{F}$-valued trace of $\mathcal{R}.$
\begin{example}\label{non-existence of trace for 2 generetors}
    Consider the $\mathbb{F}_2$-algebra  $\mathcal{R}=\mathbb{F}_2[x,y]/\langle x^2, y^2, xy\rangle.$ Let $u$ denote $x+\langle x^2, y^2, xy\rangle$ and $v$ denote $y+\langle x^2, y^2, xy\rangle$. Then any $\mathbb{F}_2$-linear map $\tau: \mathcal{R}\to \mathbb{F}_2$ is of the form:
    $$
    \tau(a+bu+cv)=\alpha a+\beta b+ \gamma c, 
    $$
    for some $\alpha, \beta, \gamma \in \mathbb{F}_2$. If one of $\beta$ or $\gamma$ is non-zero, set $\bm{s}=\gamma u+\beta v\in \mathcal{R}.$ Then
    \begin{equation*}
        \tau\left(\bm{rs}\right)=\tau\left(a\gamma u+a\beta v\right)
        =a\beta\gamma+a\beta\gamma
        =0, \forall\,\bm{r}=a+bu+cv\in \mathcal{R}.
    \end{equation*}
    If both $\gamma$ and $\beta$ are zero, then set $\bm{s}=u$ and note that $\tau(\bm{rs})=0$ for all $\bm{r}\in\mathcal{R}.$ 
    Hence, for any non-zero $\tau$, there is $\bm{s}\in \mathcal{R}\setminus\{\bm{0}\}$ such that $\tau(\bm{rs})=0,\,  \forall\,\bm{r}\in \mathcal{R},$ proving that there is no $\mathbb{F}_2$-valued trace of $\mathcal{R}.$
\end{example}
We note from the following example that writing elements of $\mathcal{R}=\mathbb{F}[x]/\langle g(x)^r\rangle$ in a special form as described in Theorem \ref{maintheorem} is crucial to determine an $\mathbb{F}$-valued trace of $\mathcal{R}.$
\begin{example}
    Let $\mathcal{R}=\mathbb{F}_2[x]/\langle (1+x)^2\rangle$ and let $u=x+\langle (1+x)^2\rangle.$ Then $\sigma:\mathcal{R}\to \mathbb{F}_2$ given by $\sigma(a+bu)=a+b$ for $a,b\in \mathbb{F}_2$ is not an $\mathbb{F}_2$-valued trace of $\mathcal{R}$ as $\ker(\sigma)=\langle 1+u\rangle.$ But  by Theorem \ref{maintheorem}, for $a,b \in \mathbb{F}_2,$ the map $\tau:\mathcal{R}\to \mathbb{F}_2$ given by by $\tau(a+bu)=\tau(a+b+b(1+u))=a+b+b=a$ is an $\mathbb{F}_2$-valued trace of $\mathcal{R}$.
\end{example}
\begin{example}\label{trace on R_2}
    Consider the $\mathbb{F}_2$-algebra $\mathcal{R}=\mathbb{F}_2[x]/\langle x^3-x\rangle\cong \mathbb{F}_2\times \mathbb{F}_2[x]/\langle (1+x)^2\rangle,$ where the isomorphism is given by:
    $$
    a+bx+cx^2+ \langle x^3-x\rangle\mapsto \left(a, a+b+c+b(1+x)+\langle (1+x)^2\rangle\right).
    $$
    If $u$ denotes $x+\langle x^3-x\rangle,$ then
    \begin{align*}
        \tau:\mathcal{R}&\to\mathbb{F}_2\\
        a+bu+cu^2&\mapsto a+(a+b+c+b)=c
    \end{align*}
    is an $\mathbb{F}_2$-valued trace of $\mathcal{R}$ by Theorem \ref{maintheorem} and Theorem \ref{existence of trace}.
\end{example}

\begin{theorem}\label{main proposition}
    Let $\mathcal{R}$ and $\mathcal{S}$ be finite-dimensional commutative $\mathbb{F}$-algebras. If both $\mathcal{R}$ and $\mathcal{S}$ admit $\mathbb{F}$-valued trace, then so does $\mathcal{R}\otimes_\mathbb{F} \mathcal{S}.$
\end{theorem}
    \begin{proof}
        Let $\tau_\mathcal{R}$ and $\tau_\mathcal{S}$ be $\mathbb{F}$-valued trace maps of $\mathcal{R}$ and $\mathcal{S}$ respectively. Suppose $\mathcal{B}=\{\bm{\alpha}_1<\bm{\alpha}_2<\dots<\bm{\alpha}_m\}$ and $\mathcal{B}^{'}=\{\bm{\beta}_1<\bm{\beta}_2\dots<\bm{\beta}_n\}$ are ordered bases for $\mathcal{R}$ and $\mathcal{S}$ respectively over $\mathbb{F}$ so that $\{\bm{\alpha}_i\otimes\bm{\beta}_j: 1\leq i\leq m, 1\leq j\leq n\}$ is a basis for $\mathcal{R}\otimes_\mathbb{F}\mathcal{S}$ over $\mathbb{F}.$\\
        Let $\bm{x}=x_1\bm{\alpha}_1+\cdots+x_m\bm{\alpha}_m\in \mathcal{R}$ and let $[\bm{x}]_{\mathcal{B}}$ be the coordinate vector of $\bm{x}$ relative to $\mathcal{B}.$ We often write $[\bm{x}]_{\mathcal{B}}$ as $\mathbf{x}=
        \begin{bmatrix}
            x_1 & x_2 & \cdots & x_m
        \end{bmatrix}^t.$ Define $\mathfrak{m}_{\bm{\alpha}_i}:\mathcal{R}\to \mathcal{R}$ by $\mathfrak{m}_{\bm{\alpha}_i}(\bm{x})=\bm{\alpha}_i\bm{x}.$ Then $\mathfrak{m}_{\bm{\alpha}_i}$ is an $\mathbb{F}$-linear map. Denote the matrix of $\mathfrak{m}_{\bm{\alpha}_i}$ relative to $\mathcal{B}$ by $A_i.$ Then $[\bm{\alpha}_i\bm{x}]_{\mathcal{B}}=A_i\mathbf{x}.$ Thus,
        \begin{equation*}
            \tau_{\mathcal{R}}(\bm{x})=\sum_{i=1}^{m}x_i\tau_{\mathcal{R}}(\bm{\alpha}_i)=\mathbf{x}^t\bm{\tau}_\mathcal{R}, \text{where}\hspace{0.1cm} \bm{\tau}_\mathcal{R}=
            \begin{bmatrix}
                \tau_{\mathcal{R}}(\bm{\alpha}_1) & \tau_{\mathcal{R}}(\bm{\alpha}_2) &
                \cdots &
                \tau_{\mathcal{R}}(\bm{\alpha}_m)
            \end{bmatrix}^t
        \end{equation*}
        and consequently, $\tau_{\mathcal{R}}(\bm{\alpha}_i\bm{x})=\mathbf{x}^tA_i^t\bm{\tau}_\mathcal{R}.$ Since $\tau_{\mathcal{R}}$ is an $\mathbb{F}$-valued trace of $\mathcal{R},$
        \begin{equation}\label{trace on R}
            \tau_{\mathcal{R}}(\bm{\alpha}_i\bm{x})=0, \forall\, 1\leq i\leq m \implies \mathbf{x}=\bm{0}.
        \end{equation}
        Let $\bm{y}=y_1\bm{\beta}_1+\cdots+y_n\bm{\beta}_n\in \mathcal{S}$ and $[\bm{y}]_{\mathcal{B}^{'}}$ be the coordinate vector of $\bm{y}$ relative to $\mathcal{B}^{'}.$ If $B_j$ is the matrix of the $\mathbb{F}$-linear map $\mathfrak{M}_{\bm{\beta}_j}:\mathcal{S}\to \mathcal{S}$ defined by $\mathfrak{M}_{\bm{\beta}_j}(\bm{y})=\bm{\beta}_j\bm{y}$ relative to $\mathcal{B}^{'}$ and writing $[\bm{y}]_{\mathcal{B}^{'}}$ as $\mathbf{y}=
        \begin{bmatrix}
            y_1 & y_2 & \cdots & y_n
        \end{bmatrix}^t,$ then by a similar argument as above, we obtain $\tau_{\mathcal{S}}(\bm{\beta}_j\bm{y})=\mathbf{y}^tB_j^t\bm{\tau}_{\mathcal{S}},$ where $\bm{\tau}_\mathcal{S}=
            \begin{bmatrix}
                \tau_{\mathcal{S}}(\bm{\beta}_1) &
                \tau_{\mathcal{S}}(\bm{\beta}_2) &
                \cdots &
                \tau_{\mathcal{S}}(\bm{\beta}_n)
            \end{bmatrix}^t,$ and,
        \begin{equation}\label{trace on S}
            \tau_{\mathcal{S}}(\bm{\beta}_j\bm{y})=0, \forall\, 1\leq j\leq n \implies \mathbf{y}=\bm{0}.
        \end{equation}
        Define
        \begin{align*}
            T:\mathcal{R}\otimes_\mathbb{F}\mathcal{S}&\to \mathbb{F}\\
            \bm{z}:=\sum_{i=1}^{m}\sum_{j=1}^{n}z_{i,j} \bm{\alpha}_i\otimes\bm{\beta}_j&\mapsto \sum_{i=1}^{m}\sum_{j=1}^{n}z_{i,j} \tau_{\mathcal{R}}(\bm{\alpha}_i)\tau_{\mathcal{S}}(\bm{\beta}_j).
        \end{align*}
        Then,
        \begin{align*}
            T(\bm{z})&=\sum_{j=1}^{n}\left(\sum_{i=1}^{m}z_{i,j} \tau_{\mathcal{R}}(\bm{\alpha}_i)\right)\tau_{\mathcal{S}}(\bm{\beta}_j)\\
            &=\sum_{j=1}^{n}\left(\mathbf{z}_{\bm{\cdot} j}^t\bm{\tau}_\mathcal{R}\right)\tau_{\mathcal{S}}(\bm{\beta}_j) &&\text{where}\hspace{0.1cm} \mathbf{z}_{\bm{\cdot} j}= \begin{bmatrix}
                z_{1,j} &
                z_{2,j} &
                \cdots &
                z_{m,j}
            \end{bmatrix}^t\\
            &=\sum_{j=1}^{n}a_j\tau_{\mathcal{S}}(\bm{\beta}_j) &&\text{where}\hspace{0.1cm}a_j=\mathbf{z}_{\bm{\cdot} j}^t\bm{\tau}_\mathcal{R}\\
            &=\bm{a}^t\bm{\tau}_{\mathcal{S}} &&\text{where}\hspace{0.1cm} \bm{a}=\begin{bmatrix}
                a_1 &
                a_2 &
                \cdots &
                a_n
            \end{bmatrix}^t
        \end{align*}
        Thus, if $\mathbf{z}=[z_{i,j}]_{m\times n},$ then 
            \begin{equation}\label{expression for T}
                T(\bm{z})=\bm{\tau}_\mathcal{R}^t \mathbf{z}\bm{\tau}_{\mathcal{S}}\,.
            \end{equation}
        Let $\mathbf{e}_k^{(r)}= \begin{bmatrix}
            0 & 0 & \cdots & 1 & \cdots & 0
        \end{bmatrix}$ denote the row vector of size $r$ whose $k$-th component is $1$ and all other components are $0$. Then $\mathbf{e}_k^{(m)}A_i$ is the $k$th row of $A_i$ and $B_j{\mathbf{e}_k^{(n)}}^t$ is the $k$th column of $B_j.$ Now, for $1\leq r\leq m,$
        \begin{align*}
            \bm{\alpha}_r\bm{z}&=\sum_{j=1}^{n}\left(\sum_{i=1}^{m}z_{i,j}\bm{\alpha}_i\bm{\alpha}_r\right)\otimes\bm{\beta}_j\\
            &=\sum_{j=1}^{n}\sum_{i=1}^{m}\mathbf{e}_i^{(m)}A_r\mathbf{z}_{\bm{\cdot} j}(\bm{\alpha}_i\otimes\bm{\beta}_j)
        \end{align*}
        Similarly, for $1\leq s\leq n,$ we have
        \begin{align*}
            \bm{\beta}_s\bm{z}=\sum_{j=1}^{n}\sum_{i=1}^{m}\mathbf{e}_j^{(n)}B_s\mathbf{z}_{i \bm{\cdot}}^t(\bm{\alpha}_i\otimes\bm{\beta}_j)
        \end{align*}
        where $\mathbf{z}_{i \bm{\cdot}}=
        \begin{bmatrix}
            z_{i, 1} & z_{i, 2} & \cdots & z_{i, n}
        \end{bmatrix}$ and consequently
        \begin{equation*}
            \bm{\beta}_s\bm{\alpha}_r\bm{z}=\sum_{j=1}^{n}\sum_{i=1}^{m}\mathbf{e}_j^{(n)}B_s\mathbf{w}_{i \bm{\cdot}}^t(\bm{\alpha}_i\otimes\bm{\beta}_j)
        \end{equation*}
        where
        \begin{align*}
            \mathbf{w}_{i \bm{\cdot}}^t&=\begin{bmatrix}
                \mathbf{e}_i^{(m)}A_r\mathbf{z}_{\bm{\cdot} 1} & \mathbf{e}_i^{(m)}A_r\mathbf{z}_{\bm{\cdot} 2} & \cdots & \mathbf{e}_i^{(m)}A_r\mathbf{z}_{\bm{\cdot} n}
            \end{bmatrix}^t\\
            &=\left(\mathbf{e}_i^{(m)}A_r\mathbf{z}\right)^t\\
            &=\mathbf{z}^t A_r^t {\mathbf{e}_i^{(m)}}^t
        \end{align*}
        Hence,
    \begin{equation*}
        \bm{\beta}_s\bm{\alpha}_r\bm{z}=\sum_{j=1}^{n}\sum_{i=1}^{m}u_{i,j}(\bm{\alpha}_i\otimes\bm{\beta}_j)
    \end{equation*}
    where $u_{i,j}=\mathbf{e}_j^{(n)}B_s\mathbf{z}^t A_r^t {\mathbf{e}_i^{(m)}}^t.$ Then by (\ref{expression for T}),
    \begin{equation}\label{Expression for T(...)}
        T(\bm{\beta}_s\bm{\alpha}_r\bm{z})=\bm{\tau}_\mathcal{R}^t \mathbf{u}\bm{\tau}_{\mathcal{S}},\,\text{where}\,\,\bm{u}=[u_{i,j}]_{m\times n}. 
    \end{equation}
    Observe, $\mathbf{u}=A_r\mathbf{z}B_s^t.$
    Hence, by (\ref{Expression for T(...)})
    \begin{equation}
        T(\bm{\beta}_s\bm{\alpha}_r\bm{z})=\left(\mathbf{z}^tA_r^t\bm{\tau}_\mathcal{R}\right)^t B_s^t\bm{\tau}_{\mathcal{S}}.
    \end{equation}
    Suppose $ T(\bm{\beta}_s\bm{\alpha}_r\bm{z})=0, \forall\, 1\leq r\leq m$ and $\forall\, 1\leq s\leq n.$ Then
    \begin{align*}
        &\left(\mathbf{z}^tA_r^t\bm{\tau}_\mathcal{R}\right)^t B_s^t\bm{\tau}_{\mathcal{S}}=0, \forall\, 1\leq s\leq n.\\
        &\implies \mathbf{z}^tA_r^t\bm{\tau}_\mathcal{R}=\mathbf{0}, \forall\, 1\leq r\leq m &&\text{(using (\ref{trace on S}))}\\
        &\implies {\mathbf{z}_{i\bm{\cdot}}}^tA_r^t\bm{\tau}_\mathcal{R}=0, \forall\, 1\leq r\leq m, \forall\,1\leq i\leq m\\
        &\implies \mathbf{z}_{i\bm{\cdot}}=\bm{0}, 1\leq i \leq m &&\text{(using (\ref{trace on R}))}\\
        &\implies \mathbf{z}=\bm{0}
    \end{align*}
    \end{proof}
    \begin{remark}
    The map $T$ is independent of the choice of the bases of $\mathcal{R}$ and $\mathcal{S}.$ For instance, if $\{\bm{\alpha}_1^{'}<\dots<\bm{\alpha}_m^{'}\}$ is another ordered basis for $\mathcal{R},$ then $T^{'}=T,$ where
\begin{equation*}
    T^{'}\left(\sum_{i=1}^{m}\sum_{j=1}^{n}z_{i,j} \bm{\alpha}_i^{'}\otimes\bm{\beta}_j\right)= \sum_{i=1}^{m}\sum_{j=1}^{n}z_{i,j} \tau_{\mathcal{R}}(\bm{\alpha}_i^{'})\tau_{\mathcal{S}}(\bm{\beta}_j).
\end{equation*}
Write $\bm{\alpha}_i^{'}=\sum_{k=1}^{m}p_k^{(i)}\bm{\alpha}_k.$ Then for $1\le i\le m, 1\le j\le n,$
\begin{align*}
    T^{'}( \bm{\alpha}_i^{'}\otimes\bm{\beta}_j)&=\tau_{\mathcal{R}}(\bm{\alpha}_i^{'})\tau_{\mathcal{S}}(\bm{\beta}_j)\\
    &= \tau_{\mathcal{R}}\left(\sum_{k=1}^{m}p_k^{(i)}\bm{\alpha}_k\right)\tau_{\mathcal{S}}(\bm{\beta}_j)\\
    &=\sum_{k=1}^{m}p_k^{(i)}\tau_{\mathcal{R}}(\bm{\alpha}_k)\tau_{\mathcal{S}}(\bm{\beta}_j)\\
    &=T\left(\sum_{k=1}^{m}p_k^{(i)} \bm{\alpha}_k\otimes\bm{\beta}_j\right)\\
    &=T(\bm{\alpha}_i^{'}\otimes\bm{\beta}_j)
\end{align*}
so that $T^{'}=T.$
\end{remark}
\par
    From the proof of Theorem \ref{main proposition}, we record the explicit form of an $\mathbb{F}$-valued trace of $\mathcal{R}\otimes_{\mathbb{F}}\mathcal{S}$ in the following corollary.
\begin{corollary}\label{main corollary}
    Let $\mathcal{R}$ and $\mathcal{S}$ be finite-dimensional commutative $\mathbb{F}$-algebras. If $\{\bm{\alpha}_1,\dots,\bm{\alpha}_m\}$ and $\{\bm{\beta}_1,\dots,\bm{\beta}_n\}$ are bases for $\mathcal{R}$ and $\mathcal{S}$ respectively over $\mathbb{F},$ then $T\left(\sum_{i=1}^{m}\sum_{j=1}^{n}z_{i,j} \bm{\alpha}_i\otimes\bm{\beta}_j\right)= \sum_{i=1}^{m}\sum_{j=1}^{n}z_{i,j} \tau_{\mathcal{R}}(\bm{\alpha}_i)\tau_{\mathcal{S}}(\bm{\beta}_j)$ defines an $\mathbb{F}$-valued trace of $\mathcal{R}\otimes_{\mathbb{F}}\mathcal{S},$ where $\tau_{\mathcal{R}}$ and $\tau_{\mathcal{S}}$ are $\mathbb{F}$-valued traces of $\mathcal{R}$ and $\mathcal{S}$ respectively.
\end{corollary}
By induction we have:
\begin{corollary}\label{existence of trace on tensor product}
    Let $\mathcal{R}_1,\mathcal{R}_2,\dots,\mathcal{R}_n$ be finite-dimensional commutative $\mathbb{F}$-algebras. If there is an $\mathbb{F}$-valued trace of each $\mathcal{R}_i,$ then there is an $\mathbb{F}$-valued trace of $\mathcal{R}_1\otimes_{\mathbb{F}}\mathcal{R}_2\otimes_{\mathbb{F}}\cdots\otimes_{\mathbb{F}}\mathcal{R}_n.$
\end{corollary}
\begin{corollary}\label{most general case where trace exists}
    If $\mathcal{R}=\mathbb{F}[x_1, x_2,\dots, x_n]/\langle g_1(x_1), g_2(x_2),\dots, g_n(x_n)\rangle,$ where $g_i(x_i)\in \mathbb{F}[x_i]\setminus\{\bm{0}\},$ then $\mathcal{R}$ admits an $\mathbb{F}$-valued trace.
\end{corollary}
\begin{proof}
    By Theorem \ref{tensorproduct}, $\mathcal{R}\cong\mathbb{F}[x_1]/\langle g_1(x_1)\rangle\otimes_{\mathbb{F}}\mathbb{F}[x_2]/\langle g_2(x_2)\rangle\otimes_{\mathbb{F}}\cdots\otimes_{\mathbb{F}}\mathbb{F}[x_n]/\langle g_n(x_n)\rangle.$ By Theorem \ref{existence of trace}, there is an $\mathbb{F}$-valued trace of each $\mathbb{F}[x_i]/\langle g_i(x_i)\rangle,$ hence, by Corollary \ref{existence of trace on tensor product}, there is an $\mathbb{F}$-valued trace of $\mathcal{R}.$  
\end{proof}

\begin{example}
    Let $\mathcal{R}=\mathbb{F}[x,y]/\langle x^2, y^2\rangle$ and set $w=x+\langle x^2 \rangle.$ Then $\tau: \mathbb{F}[x]/\langle x^2\rangle \to \mathbb{F}$ defined by $\tau(a+bw)=a+b$ is an $\mathbb{F}$-valued trace on $\mathbb{F}[x]/\langle x^2\rangle$ by Theorem \ref{maintheorem} and hence $T:\mathcal{R}\to \mathbb{F}$ defined by $T(a+bu+cv+duv)=a+b+c+d$ is an $\mathbb{F}$-valued trace on $\mathcal{R}$ by Corollary \ref{main corollary}, where $u=x+\langle x^2, y^2\rangle$ and $v=y+\langle x^2, y^2\rangle.$
\end{example}

\section{Applications}\label{Section 4}
\subsection{Bases and dual basis}
    Let $\mathcal{R}$ be a finite-dimensional commutative $\mathbb{F}$-algebra that admits an $\mathbb{F}$-valued trace $\tau.$ For $\bm{\beta}\in \mathcal{R},$ define
    \begin{align*}
        f_{\bm{\beta}}: \mathcal{R}&\to \mathbb{F}\\
        f_{\bm{\beta}}(\bm{\alpha})& =\tau(\bm{\beta\alpha}),\,\forall\, \bm{\alpha}\in \mathcal{R}.
    \end{align*}
    By $\mathbb{F}$-linearity of $\tau,$ $f_{\bm{\beta}}$ is a linear functional on $\mathcal{R}.$ Moreover, we have
\begin{theorem}\label{trace map and linear map}
    If $\mathcal{R}$ is a finite-dimensional commutative $\mathbb{F}$-algebra that admits an $\mathbb{F}$-valued trace $\tau,$ then for every functional $f:\mathcal{R}\to \mathbb{F},$ there exists $\bm{\beta}\in \mathcal{R}$ such that $f=f_{\bm{\beta}}.$
\end{theorem}
\begin{proof}
    Let $\bm{\beta}, \bm{\gamma}\in \mathcal{R},$ with $\bm{\beta}\neq\bm{\gamma}.$ Since $\bm{\beta}-\bm{\gamma}\neq\bm{0},$ it follows from the definition of $\tau$ that there is some $\bm{\alpha}\in \mathcal{R}$ such that $\tau((\bm{\beta}-\bm{\gamma})\bm{\alpha})\neq 0.$ Hence, $ f_{\bm{\beta}}(\bm{\alpha})-f_{\bm{\gamma}}(\bm{\alpha})=\tau((\bm{\beta}-\bm{\gamma})\bm{\alpha})\neq 0,$ and so the maps $f_{\bm{\beta}}$ and $f_{\bm{\gamma}}$ are different. 
\end{proof}

Let $\mathcal{R}$ be a commutative $\mathbb{F}$-algebra of dimension $n$ that admits an $\mathbb{F}$-valued trace $\tau.$ If $\{\bm{\alpha}_1,\dots, \bm{\alpha}_n\}$ is a basis for $\mathcal{R}$ over $\mathbb{F},$ then the projection map $\pi_j:\mathcal{R}\to \mathbb{F}, \sum_{i=1}^{n}x_i\bm{\alpha}_j\mapsto x_j$ is an $\mathbb{F}$-linear map from $\mathcal{R}$ to $\mathbb{F},$ hence by Theorem \ref{trace map and linear map}, there is $\bm{\beta}_j\in \mathcal{R}$ such that $\pi_j(\bm{\alpha})=\tau(\bm{\beta}_j\bm{\alpha})$ for all $\bm{\alpha}\in \mathcal{R}.$ Putting $\bm{\alpha}=\bm{\alpha}_i, 1\leq i\leq n,$ we have \[ \tau(\bm{\beta}_j\bm{\alpha}_i)=\begin{cases} 
      0 & \text{if}\,\, i\neq j \\
      1 & \text{if}\,\, i=j 
   \end{cases}
\]Thus, $\{\beta_1,\dots,\beta_n\}$ is also a basis for $\mathcal{R}$ over $\mathbb{F}.$
\begin{definition}
    Let $\mathcal{R}$ be a finite-dimensional commutative $\mathbb{F}$-algebra that admits an $\mathbb{F}$-valued trace $\tau.$ Then the two bases $\{\bm{\alpha}_1,\dots,\bm{\alpha}_n\}$ and $\{\bm{\beta}_1,\dots,\bm{\beta}_n\}$ for $\mathcal{R}$ over $\mathbb{F}$ are said to be \textit{$\tau$-dual (or $\tau$- complementary) bases} if $\tau(\bm{\alpha}_i \bm{\beta}_j)=\delta_{ij},$ where $\delta_{ij}$ denotes the Kronecker symbol.
\end{definition}

From the above discussion, we have the following proposition.
\begin{proposition}\label{existence of dual basis}
    Let $\mathcal{R}$ be a finite-dimensional commutative $\mathbb{F}$-algebra. If there is an $\mathbb{F}$-valued trace $\tau$ of $\mathcal{R},$ then for any basis for $\mathcal{R}$ over $\mathbb{F},$ there exists a unique $\tau$-dual basis.
\end{proposition}
\begin{example}
    Let $\mathcal{R}=\mathbb{F}[x_1, x_2,\dots, x_n]/\langle g_1(x_1), g_2(x_2),\dots, g_n(x_n)\rangle,$ where $g_i(x_i)\in \mathbb{F}[x_i].$ Then by Corollary \ref{most general case where trace exists}, there is an $\mathbb{F}$-valued trace $\tau$ of $\mathcal{R}.$ Hence, for any basis for $\mathcal{R},$ there is a $\tau$-dual basis.
\end{example}
\begin{definition}
    Let $\mathcal{R}$ be a commutative $\mathbb{F}$-algebra of dimension $n$ that admits an $\mathbb{F}$-valued trace $\tau.$ Then the discriminant $\Delta_{\tau}(\bm{\alpha}_1,\dots, \bm{\alpha}_n)$ of the elements $\bm{\alpha}_1,\dots, \bm{\alpha}_n\in \mathcal{R}$ is defined by the determinant of order $n$ given by
    \begin{equation*}
        \Delta_{\tau}(\bm{\alpha}_1,\dots, \bm{\alpha}_n)=
        \begin{vmatrix}
            \tau(\bm{\alpha}_1\bm{\alpha}_1) & \tau(\bm{\alpha}_1\bm{\alpha}_2) & \cdots & \tau(\bm{\alpha}_1\bm{\alpha}_n)\\
            \tau(\bm{\alpha}_2\bm{\alpha}_1) & \tau(\bm{\alpha}_2\bm{\alpha}_2) & \cdots & \tau(\bm{\alpha}_2\bm{\alpha}_n)\\
            \vdots & \vdots & \vdots & \vdots\\
            \tau(\bm{\alpha}_n\bm{\alpha}_1) & \tau(\bm{\alpha}_n\bm{\alpha}_2) & \cdots & \tau(\bm{\alpha}_n\bm{\alpha}_n)\\
        \end{vmatrix}
    \end{equation*}
\end{definition}
Note that $\Delta_{\tau}(\bm{\alpha}_1,\dots, \bm{\alpha}_n)\in\mathbb{F}.$ 
\par
A simple characterization of bases for $\mathcal{R}$ which admits an $\mathbb{F}$-valued trace can be given as follows:
\begin{theorem}
    Let $\mathcal{R}$ be a commutative $\mathbb{F}$-algebra of dimension $n$ that admits an $\mathbb{F}$-valued trace $\tau.$ Then $\{\bm{\alpha}_1,\dots, \bm{\alpha}_n\}$ is a basis for $\mathcal{R}$ over $\mathbb{F}$ iff $\Delta_{\tau}(\bm{\alpha}_1,\dots, \bm{\alpha}_n)\neq0.$
\end{theorem}
\begin{proof}
    Let $\{\bm{\alpha}_1,\dots, \bm{\alpha}_n\}$ be a basis for $\mathcal{R}$ over $\mathbb{F}.$ We show that the rows of the determinant $\Delta_{\tau}(\bm{\alpha}_1,\dots, \bm{\alpha}_n)$ are linearly independent. Suppose that $c_1\tau(\bm{\alpha}_1\bm{\alpha}_j)+\cdots+c_n\tau(\bm{\alpha}_n\bm{\alpha}_j)=0$ for $1\leq j\leq n,$ where $c_i\in \mathbb{F}.$ If $\bm{\beta}=c_1\bm{\alpha}_1+\cdots+c_n\bm{\alpha}_n,$ then $\tau(\bm{\beta\alpha}_j)=0,$ for $1\leq j\leq n.$ However, this implies $\bm{\beta}=\bm{0}$ as $\{\bm{\alpha}_1,\dots, \bm{\alpha}_n\}$ forms a basis for $\mathcal{R}.$ Hence, $c_1=\cdots=c_n=0,$ which ensures that $\Delta_{\tau}(\bm{\alpha}_1,\dots, \bm{\alpha}_n)\neq0.$
    \par
    Conversely, suppose that $\Delta_{\tau}(\bm{\alpha}_1,\dots, \bm{\alpha}_n)\neq0.$ Let $c_1\bm{\alpha}_1+\cdots+c_n\bm{\alpha}_n=\bm{0},$ for some $c_i\in \mathbb{F}.$ Multiplying by $\bm{\alpha}_j$ and applying $\tau,$ we obtain
    \begin{equation*}
        c_1\tau(\bm{\alpha}_1\bm{\alpha}_j)+\cdots+c_n\tau(\bm{\alpha}_n\bm{\alpha}_j)=0\,\, \text{for}\, 1\le j\le n.
    \end{equation*}
    Since $\Delta_{\tau}(\bm{\alpha}_1,\dots, \bm{\alpha}_n)\neq0,$ it follows that $c_1=\cdots=c_n=0.$
\end{proof}

\subsection{Applications to Algebraic coding theory}
Let $R$ be a finite commutative ring. Then $R^n$ is a free $R$-module of rank $n.$ Any non-empty subset $\mathcal{C}$ of $R^n$ is called a \textit{code} of length $n$ over $R$. If, in addition, $\mathcal{C}$ is a $R$-submodule of $R^n,$ then the code $\mathcal{C}$ is called a \textit{linear code}. One may refer to \cite{pless} and \cite{sole2009codes} for more on codes over rings and fields.
\par
Let $\mathcal{R}_q$ be a finite-dimensional commutative $\mathbb{F}_q$-algebra that admits an $\mathbb{F}_q$-valued trace $\tau.$ Then the trace map $\tau$ can be used to go down from a code defined over $\mathcal{R}_q$ to a code over $\mathbb{F}_q$ as follows:

\begin{definition}[$\tau$-trace code]
    For a linear code $\mathcal{C}$ of length $n$ over $\mathcal{R}_q,$ define the $\tau$-trace code of $\mathcal{C}$ by
    \begin{equation*}
        \tau\left(\mathcal{C}\right):=\left\{\left(\tau(c_1),\dots,\tau(c_n)\right): (c_1,\dots,c_n)\in \mathcal{C} \right\},
    \end{equation*}
    which is a linear code of length $n$ over $\mathbb{F}_q.$
\end{definition}
Suppose $\varphi:\mathbb{F}_q\hookrightarrow\mathcal{R}_q$ defines the $\mathbb{F}_q$-algebra structure of $\mathcal{R}_q,$ then we identify $\mathbb{F}_q$ with $\varphi(\mathbb{F}_q).$
\par
Another way to go down from a code defined over $\mathcal{R}_q$ to a code over $\mathbb{F}_q$ is the following:
\begin{definition}[Subfield subcode]
    Let $\mathcal{C}$ be a linear code of length $n$ over $\mathcal{R}_q.$ The code $\mathcal{C}|_{\mathbb{F}_q}:=\mathcal{C}\cap\mathbb{F}_q^n,$ is called the subfield subcode.
\end{definition}
If $\Tr$ is the usual trace map from $\mathbb{F}_{q^m}$ to $\mathbb{F}_q,$ then there is a nice relationship between the trace and subfield subcodes. Interestingly, the same relationship holds between $\tau$-trace code and subfield subcodes.
\begin{theorem}
    Let $\tau$ be an $\mathbb{F}_q$-valued trace of $\mathcal{R}_q.$ Suppose $\mathcal{C}$ is a linear code of length $n$ over $\mathcal{R}_q,$ then
    \begin{equation*}
        \tau\left(\mathcal{C}^\bot\right)=\left(\mathcal{C}|_{\mathbb{F}_q}\right)^\bot
    \end{equation*}
    where $\mathcal{C}^\bot=\{\bm{x}=(x_1,\dots,x_n)\in \mathcal{R}_q^n:\bm{x}\cdot\bm{y}=\sum_{i=1}^{n}x_iy_i=0,\, \forall\, \bm{y}=(y_1,\dots,y_n)\in \mathcal{C}\}.$
\end{theorem}
\begin{proof}
Let $\bm{c}=\left(c_1, \dots, c_n\right)\in \mathcal{C}|_{\mathbb{F}_q}$ and $\bm{a}=\left(a_1, \dots, a_n\right)\in \mathcal{C}^{\bot}$. Then
\begin{equation*}
\bm{c} \cdot \tau(\bm{a})=\sum_{i=1}^n c_i \tau(a_i)=\tau\left(\sum_{i=1}^n c_i a_i\right)=\tau(\bm{c} \cdot \bm{a})=0,
\end{equation*}
showing that $\tau\left(\mathcal{C}^\bot\right)\subseteq\left(\mathcal{C}|_{\mathbb{F}_q}\right)^\bot.$
Next, we show that $\tau\left(\mathcal{C}^\bot\right)\supseteq\left(\mathcal{C}|_{\mathbb{F}_q}\right)^\bot.$ This assertion is equivalent to
$$
\left(\tau\left(\mathcal{C}^{\perp}\right)\right)^{\perp} \subseteq \mathcal{C}|_{\mathbb{F}_q} .
$$
Suppose the above relationship does not hold, then there exist some $\bm{u} \in\left(\tau\left(\mathcal{C}^{\bot}\right)\right)^{\bot} \setminus \mathcal{C}|_{\mathbb{F}_q}$
    and $\bm{v} \in \mathcal{C}^{\bot}$ with $\bm{u} \cdot \bm{v} \neq 0$. As $\tau$ is an $\mathbb{F}_q$-valued trace of $\mathcal{R}_q$, there is an element $\alpha \in \mathcal{R}_q$ such that $\tau(\alpha(\bm{u} \cdot \bm{v})) \neq 0$. Hence,
$$
\bm{u} \cdot \tau(\alpha \bm{v})=\tau(\bm{u} \cdot \alpha \bm{v})=\tau(\alpha(\bm{u} \cdot \bm{v})) \neq 0 .
$$

But, on the other hand, we have $\bm{u} \cdot \tau(\alpha \bm{v})=0$ because $\bm{u} \in$ $\left(\tau\left(\mathcal{C}^{\bot}\right)\right)^{\bot}$ and $\alpha \bm{v} \in C^{\bot}$. The desired result follows from this contradiction.
\end{proof}
An $\mathbb{F}_q$-valued trace $\tau$ of $\mathcal{R}_q$ can also be used as a tool to construct linear codes over $\mathbb{F}_q$ as follows:\\
Let $D=\{\{\bm{d}_1< \bm{d}_2<...<\bm{d}_n\}\}$ be an ordered multiset, where each $\bm{d}_i\in \mathcal{R}_q.$ Define
$$
\mathcal{C}_D:=\{\left(\tau(\bm{xd}_1), \tau(\bm{xd}_2),\dots, \tau(\bm{xd}_n)\right): \bm{x}\in \mathcal{R}_q\}
$$
Then $\mathcal{C}_D$ is a linear code of length $n$ over $\mathbb{F}_q$ and we call $D$ the \emph{defining sequence} of the code $\mathcal{C}_D.$
\begin{example}
    Consider $\mathcal{R}_2=\mathbb{F}_2[x]/\langle x^3-x\rangle$ and let $u=x+\langle x^3-x\rangle.$ Then by Example \ref{trace on R_2}, $\tau(a+bu+cu^2)=c$ is an $\mathbb{F}_2$-valued trace of $\mathcal{R}_2.$ Let $D=\{\{1<u<1+u<1+u^2<u+u^2<1+u+u^2\}\}.$ Then $\mathcal{C}_D$ is a binary $[6,3,3]$-\textit{quasicyclic} linear code of degree $2.$
\end{example}
\subsubsection{Construction of subfield codes}
An $\mathbb{F}_q$-valued trace $\tau$ of $\mathcal{R}_q$ is useful in the computation of subfield codes.\\
Suppose that $\mathcal{C}$ is a code of length $n$ over $\mathcal{R}_q$ generated by the (full rank) matrix $G$ and $\mathcal{B}$ be a basis for $\mathcal{R}_q$ over $\mathbb{F}_q.$ The code $\mathcal{C}^{(q)}$ over $\mathbb{F}_q$ generated by the matrix which is obtained by replacing each entry of $G$ by its column representation relative to $\mathcal{B},$ is called the \textit{subfield code} of $\mathcal{C}$\cite{ding2019subfield}. In fact, it is independent of the choice of $\mathcal{B}.$

Let $\{\bm{\alpha}_1,\dots,\bm{\alpha}_m\}$ be a basis for $\mathcal{R}_q$ and let $\{\bm{\beta}_1,\dots,\bm{\beta}_m\}$ be its $\tau$-dual basis for $\mathcal{R}_q$ over $\mathbb{F}_q$. Then if $\bm{r}=\sum_{i=1}^{m}r_i\bm{\alpha}_i \in \mathcal{R}_q,$ for $r_i\in \mathbb{F}_q,$ then
$$
r_i=\tau(\bm{r}\bm{\beta}_i).
$$
With the above discussion, Theorem 2.4 of \cite{ding2019subfield} gets generalized to $\mathcal{R}_q.$
\begin{theorem}
    Let $\tau:\mathcal{R}_q\to \mathbb{F}_q$ be an $\mathbb{F}_q$-valued trace of $\mathcal{R}_q$ and let $\mathcal{B}=\{\bm{\alpha}_1,\dots,\bm{\alpha}_m\}$ be a basis for $\mathcal{R}_q$ over $\mathbb{F}_q.$ Suppose that $\mathcal{C}$ is a linear code of length $n$ over $\mathcal{R}_q$ generated by:
    $$
    G=\begin{bmatrix}
        \bm{g}_{11} & \bm{g}_{12} & \cdots & \bm{g}_{1n} \\
        \bm{g}_{21} & \bm{g}_{22} & \cdots & \bm{g}_{2n} \\
        \vdots  & \vdots & \vdots & \vdots\\
        \bm{g}_{k1} & \bm{g}_{k2} & \cdots & \bm{g}_{kn} \\
    \end{bmatrix}$$
    Then the subfield code $\mathcal{C}^{(q)}$ of $\mathcal{C}$ is generated by
    $$
        G^{(q)}=\begin{bmatrix}
            G_1^{(q)}\\
            G_2^{(q)}\\
            \vdots\\
            G_k^{(q)}
        \end{bmatrix}
    $$
    where for $1\leq i\leq k,$
    $$
        G_i^{(q)}=\begin{bmatrix}
            \tau(\bm{g}_{i1}\bm{\alpha}_1) & \tau(\bm{g}_{i2}\bm{\alpha}_1) & \dots & \tau(\bm{g}_{in}\bm{\alpha}_1)\\
            \tau(\bm{g}_{i1}\bm{\alpha}_2) & \tau(\bm{g}_{i2}\bm{\alpha}_2) & \dots & \tau(\bm{g}_{in}\bm{\alpha}_2)\\
            \vdots & \vdots & \vdots & \vdots\\
            \tau(\bm{g}_{i1}\bm{\alpha}_m) & \tau(\bm{g}_{i2}\bm{\alpha}_m) & \dots & \tau(\bm{g}_{in}\bm{\alpha}_m)\\
        \end{bmatrix}.
    $$
\end{theorem}

\section{Conclusion and Discussion}\label{Section 5}

In this manuscript, we studied the $\mathbb{F}$-valued trace of a finite-dimensional commutative $\mathbb{F}$-algebra. Given a finite-dimensional commutative $\mathbb{F}$-algebra, an $\mathbb{F}$-valued trace may not exist; however, we proved that a finite-dimensional commutative $\mathbb{F}$-algebra of the form $\mathcal{R}=\mathbb{F}[x_1, x_2,\dots, x_n]/\langle g_1(x_1), g_2(x_2),\dots, g_n(x_n)\rangle,$ $g_i(x_i)\in \mathbb{F}[x_i],$ possesses an $\mathbb{F}$-valued trace, and we constructed such a map on $\mathcal{R}$. We showed that an $\mathbb{F}$-valued trace on a finite-dimensional commutative $\mathbb{F}$-algebra induces a non-degenerate bilinear form on $\mathcal{R}$ and hence determines all linear transformations from $\mathcal{R}$ to $\mathbb{F},$ and it is helpful in the  characterisation of bases for $\mathcal{R}.$ In the field of algebraic coding theory, we presented how an $\mathbb{F}_q$-valued trace can be used as a tool to descend from codes over an $\mathbb{F}_q$-algebra to codes over $\mathbb{F}_q.$ 
\par
A non-zero $\mathbb{F}$-valued $\mathbb{F}$-linear map on a non-commutative $\mathbb{F}$-algebra is called an $\mathbb{F}$-valued trace if its kernel contains no non-zero left ideals. Let $\mathcal{R}=M_n(\mathbb{F}),$ the ring of all $n\times n$ matrices over $\mathbb{F}.$ Then $\mathcal{R}$ is a non-commutative $\mathbb{F}$-algebra for $n\ge 2$. Define $\tau:\mathcal{R}\to\mathbb{F}$ by $A=[a_{ij}]\mapsto\sum_{i=1}^{n}a_{ii},$ the usual trace of $A.$ Let $\mathcal{I}$ be any non-zero left ideal of $\mathcal{R}$ and suppose that $B=(b_{ij})\in \mathcal{I}$ be such $b_{ij}\neq 0$ for some $1\le i, j\le n$. If $E_{ij}\in \mathcal{R}$ be such that whose $(i,j)$th entry is $1$ and whose other entries are all $0,$ then $\Tr(E_{jj}\sum_{k=1}^{n}E_{ki}B)=b_{ij}\neq 0,$ and consequently, $\mathcal{I}\nsubseteq \ker(\tau).$ Hence, $\ker(\tau)$ contains no non-zero left ideals of $\mathcal{R},$ proving that $\tau$ is an $\mathbb{F}$-valued trace of $\mathcal{R}.$
\par
The non-commutative $\mathbb{F}_2$-algebra\, $\mathcal{U}=\left\{A=\begin{bmatrix}
    a & b\\
    0 & c
\end{bmatrix}: a,b,c\in \mathbb{F}_2\right\}$ does not admit any $\mathbb{F}_2$-valued trace map.
\par
Consider the non-commutative non-unital ring $E=\langle a,b\,|\, 2a=2b=0, a^2=a, b^2=b, ab=b, ba=a\rangle.$ Note that $E$ has no $\mathbb{F}_2$-algebra structure and the underlying set of $E$ is $\{0, a, b, c=a+b\}.$ Consider the following action of $\mathbb{F}_2$ on $E$: $0e=e0=0$ and $1e=e1=e$ for all $e\in E$. Then every element of $E$ can be expressed as $as+ct$ for $s,t\in \mathbb{F}_2.$ The only non-zero projections of $E$ onto $\mathbb{F}_2$ are $\tau_i:E\to \mathbb{F}_2, i=1,2,3,$ where
    \begin{align*}
        \tau_1(as+ct)&=a\\
        \tau_2(as+ct)&=c\\
        \tau_3(as+ct)&=a+c.
    \end{align*}
    It is not difficult to verify that $\ker(\tau_i)$ contains non-zero left ideals of $E,$ for $i=1, 2, 3.$
\section*{Declarations}
\subsection*{Conflict of Interest}
Both authors declare that they have no conflict of interest.

\section*{Acknowledgements}
The work of the first author was supported by Council of Scientific and Industrial Research (CSIR) India, under the grant no. 09/0086(13310)/2022-EMR-I.

\bibliographystyle{abbrv}
\bibliography{TraceAlgebra}

\end{document}